\documentclass[a4paper,12pt]{article}

\title{Quasi-cyclic codes of index 2}
\author{Kanat Abdukhalikov \\	
Department of Mathematical Sciences, \\
UAE University, PO Box 15551, Al Ain, UAE\\
Email: abdukhalik@uaeu.ac.ae \\
Askar S. Dzhumadil'daev \\
Institute of Mathematics and Mathematical Modeling, Almaty, Kazakhstan\\
Email: dzhuma@hotmail.com \\
San Ling \\
School of Physical and Mathematical Sciences, \\
Nanyang Technological University, Singapore 637371\\
 E-mail: lingsan@ntu.edu.sg
}

\date{ }

\usepackage{amsthm,amsmath,amssymb} 
\usepackage{url} 

\begin{document} 

\maketitle

\theoremstyle{plain} 
\newtheorem{lemma}{Lemma}[section] 
\newtheorem{theorem}[lemma]{Theorem}
\newtheorem{corollary}[lemma]{Corollary}
\newtheorem{proposition}[lemma]{Proposition}

\theoremstyle{definition} 
\newtheorem{definition}{Definition}[section] 
\newtheorem{remark}{Remark}[section] 
\newtheorem{example}{Example}

\newcommand{\eps}{\varepsilon}
\newcommand{\inprod}[1]{\left\langle #1 \right\rangle}
\newcommand{\la}{\lambda} 
\newcommand{\al}{\alpha}
\newcommand{\om}{\omega} 
\newcommand{\gam}{\gamma}

\begin{abstract} 
We study quasi-cyclic codes of index 2 over finite fields. 
We give a classification of such codes.
Their duals with respect to the Euclidean, symplectic and Hermitian 
inner products are investigated. 
We describe self-orthogonal and dual-containing codes. 
Lower bounds for minimum distances of quasi-cyclic codes are given. 
A quasi-cyclic code of index 2  is generated  by at most two elements. 
We describe conditions when such a code (or its dual) is generated by one element. 
\end{abstract}
\text{\bf Keywords:} \small{Quasi-cyclic codes, dual codes,  lower bounds, quantum error-correcting codes.}\\
\text{\bf Mathematics Subject Classifications(2010):} \small{ 94B05, 94B15, 94B60.}\\


\section{Introduction}

Quasi-cyclic (QC) codes are an important type of linear codes which is a generalization of cyclic codes. 
There are families of QC codes that are asymptotically good \cite{Kasami,Ling2003-2}. Many good codes have been 
constructed with a better minimum distance than any linear code of the same length and dimension previously constructed. 

Ling and Sol\'e explored the algebraic structure of QC codes in-depth in a series of articles \cite{Ling2001,Ling2003,Ling2005,Ling2006}. 
In \cite{Conan,Lally} the algebraic structure of quasi-cyclic codes is discussed with a different approach. 
A class of 1-generator quasi-cyclic codes and their properties have been studied in \cite{Seguin} and \cite{Seguin1990}. 
A generalization -- the structure of 1-generator quasi-twisted codes -- was investigated by Aydin {\it et al.} \cite{Aydin},     
where they also constructed some new linear codes. Esmaeili and Yari \cite{Esmaeili} provided a criterion for generalized QC codes to be one-generator, however application of this criterion requires knowledge of the decomposition of codes into a sum of component codes (which requires the additional condition $\gcd(p,m)=1$) and looking at their properties. 

Dastbasteh and Shivji \cite{Dastbasteh} considered additive cyclic codes over ${\mathbb F}_{p^2}$ that can be considered as an interpretation of QC codes over ${\mathbb F}_p$. They computed the symplectic dual of additive codes by decomposing them into  components. In our paper, we study Euclidean, Hermitian, and symplectic duals without using decomposition of codes, using  only generators, and without restrictions to the characteristic of the ground field.

In \cite{Galindo,Liu,Lu,Lv2019,Lv2020,Lv2020-2,Guan2022,Guan2022-2,Guan2022-3} some particular cases of QC codes are considered.

Lally \cite{Lally2003} presented lower bounds for the minimum distance of quasi-cyclic codes.  
Semenov and Trifonov \cite{Semenov} introduced a spectral method for quasi-cyclic codes and obtained a BCH-like bound. 
This has led to several further works on distance bounds, known as spectral bounds, for quasi-cyclic and quasi-twisted codes, e.g., 
\cite{Ezerman2019,Ezerman2021,Luo,Roth,Zeh}. 
Yet another approach, using the concatenated structure of quasi-cyclic codes, also yields a different type of distance lower 
bound for quasi-cyclic codes, e.g., \cite{Guneri}.  

Other than the algebraic structure and distance bounds, some other problems related to quasi-cyclic codes have 
also been studied. For example, 
in \cite{Shi2023} QC perfect codes in Doob graphs and special partitions of Galois rings are considered. 
For some classes of quasi-cyclic codes,  sufficient conditions for self-orthogonality with respect to the symplectic, Euclidean 
and Hermitian inner products are given in \cite{Galindo}. 

This last problem is linked to the fact that, in recent years, 
quasi-cyclic codes have gained attention due to their applications in quantum error-correcting coding.  
A quantum error-correcting code (or just quantum code) is a code that protects quantum information from corruption 
by noise (decoherence) on the quantum channel in a way that is similar to how a classical error-correcting code 
protects information on the classical channel.  
 In 1995 Shor \cite{Shor} presented the first scheme to reduce errors in a quantum computer. Later, two independent 
research groups Calderbank and Shor \cite{Calderbank1996}, and Steane \cite{Steane1996,Steane1996-2} discovered 
a relationship between quantum codes and classical binary error-correcting codes. In 1998, 
Calderbank {\it et al.} \cite{Calderbank1998}  proposed a method to construct binary quantum codes using additive codes 
over the field of four elements. The binary case was generalized in 2001 by Ashikhmin and Knill \cite{Ashikhmin} 
to nonbinary quantum codes. This approach is now called the theory of stabilizer codes, which allows the construction 
of quantum error-correcting codes using classical codes that are self-orthogonal with respect to a certain 
symplectic inner product \cite{Grassl2021,Ketkar,Knill}. 
In \cite{Abdukhalikov,Galindo,Guan2022,Guan2022-2,Lv2019,Lv2020,Lv2020-2} many examples of new optimal 
quantum error-correcting codes were constructed using some particular classes of quasi-cyclic codes of index 2. 
Some other examples of application of quasi-cyclic codes to the construction of quantum codes include \cite{Ezerman2019-2} 
and \cite{Luo2024}. The latter uses quasi-cyclic codes to construct so-called entanglement-assisted quantum 
error-correcting codes with best-known parameters.

 In this paper we consider QC codes of index 2 in full generality. Unlike other papers, where the condition $\gcd(m,q)=1$ 
 is assumed, 
we do not have restrictions on the characteristic of the ground field, except for the case of one-generator  
 codes (Theorems \ref{one-gen} and \ref{dual-one-gen}). More importantly, this approach is simpler and more powerful. 
 The paper is organized as follows. We recall first in Section \ref{prelim}  definitions and notation concerning quasi-cyclic codes. 
 In Section \ref{quasi-cyclic} we give a classification of QC codes of index 2 over finite fields and 
 present lower bounds for minimum distances. 
 A QC code of index 2  is generated  by at most two elements. 
We give a necessary and sufficient condition for a QC code of index 2 to be generated by one element. 
For this, we do not need extra assumptions such as decomposition of codes into component subcodes, and we only use generators of codes.    
We describe duals of QC codes with respect to the Euclidean, symplectic and Hermitian 
inner products in Sections \ref{euclidean}, \ref{symplectic} and \ref{hermitian}, respectively, with a uniform approach (without using decomposition of codes into components and without restrictions on the characteristic of the ground field). 
We study duals of QC codes in a systematic way involving the notion of  adjoint maps. 
Moreover, we describe self-orthogonal and dual-containing QC codes. 
In Section \ref{quantum} we present a construction of quantum error correcting codes from our study.    
Finally, in Section \ref{conclusion} we summarize our results. 


\section{Preliminaries}
\label{prelim} 

 Let $F=\mathbb F_{q}$ be a finite field of $q$ elements. A linear code $C$ of length $n$ is a subspace of the vector space $F^n$. The elements of $C$ are codewords.  

Let $T$ be the standard cyclic shift operator on $F^n$. A code is said to be {\em quasi-cyclic} of index $\ell$ if it is invariant 
under $T^\ell$. We assume that $\ell$ divides $n$. If $\ell=1$ then the QC code is a cyclic code. 

Let  $R=F[x]/ \langle x^m-1 \rangle$.  We recall that cyclic codes of length $m$ over $F$ can be considered as ideals of $R$. 

Let $n=m\ell$ and let $C$ be a linear quasi-cyclic code of length $m\ell$ and index $\ell$ over $R$. Let 
$$c=(c_{0,0},c_{0,1},\dots,c_{0,\ell-1},c_{1,0},c_{1,1},\dots c_{1,\ell-1}, \dots, \linebreak  c_{m-1,0},c_{m-1,1},\dots,c_{m-1,\ell-1})$$ 
denote a codeword in $C$. 
Define a map $\varphi : F^{m\ell} \rightarrow R^{\ell}$ by 
$$\varphi (c)= (c_0(x),c_1(x),\dots, c_{\ell -1}(x)) \in R^{\ell},$$ 
where 
$$c_j(x)=c_{0,j}+c_{1,j}x+c_{2,j}x^2+\cdots+c_{m-1,j}x^{m-1}\in R.$$

The following lemma is well-known. 

\begin{lemma}[\cite{Lally,Ling2001}] 
The map $\varphi$ induces a one-to-one correspondence between quasi-cyclic codes over $F$ of index $\ell$ and 
length $m\ell$ and linear codes over $R$ of length $\ell$. 
\end{lemma}

\section{Quasi-cyclic codes of index 2}
\label{quasi-cyclic}

The next Theorem gives a description of quasi-cyclic codes  over $F$ of index 2. 

\begin{theorem}
\label{generators}
1) Let $C$ be a quasi-cyclic code of length $2m$ and index $2$. Then $C$ is generated by two elements 
 $g_1=\big( g_{11}(x),g_{12}(x)\big)$ and $g_2=\big( 0,g_{22}(x)\big)$ such that they satisfy the following conditions:  

\begin{equation}
\begin{array}{c}
\tag{$\ast$} 
g_{11}(x)\mid (x^m-1) \ {\rm and} \  g_{22}(x)\mid (x^m-1),  \\
\deg g_{12} (x) < \deg g_{22}(x),   \\
g_{11}(x)g_{22}(x) \mid  (x^m-1) g_{12}(x).
\end{array}
\end{equation}
Moreover, in this case 
$\dim C = 2m- \deg g_{11}(x) - \deg g_{22}(x)$.


2) Let the code $C$ be generated by elements $g_1=\big( g_{11}(x),g_{12}(x)\big)$ and $g_2=\big( 0,g_{22}(x)\big)$, 
and let $C'$ be generated by elements $g'_1=\big( g'_{11}(x),g'_{12}(x)\big)$ and $g'_2=\big( 0,g'_{22}(x)\big)$, 
both satisfying Conditions $(\ast)$.  Let $g_{11}(x)$, $g_{22}(x)$, $g'_{11}(x)$, $g'_{22}(x)$ be monic polynomials. 
Then $C=C'$ if and only if $g_{11}(x) = g'_{11}(x)$, $g_{22}(x) = g'_{22}(x)$, and $g_{12}(x) = g'_{12}(x)$.  

\end{theorem}

\begin{proof}
1) By \cite{Lally} we can assume that $C$ is generated by two elements $g_1=\big( g_{11}(x),g_{12}(x)\big)$ and 
$g_2=\big( 0,g_{22}(x)\big)$, satisfying the conditions 
$g_{11}(x)\mid (x^m-1)$, $g_{22}(x)\mid (x^m-1)$, $\deg g_{12} (x) < \deg g_{22}(x)$, 
and $\dim C = 2m- \deg g_{11}(x) - \deg g_{22}(x)$. 
Existence of such generators are equivalent \cite{Lally} to the existence of a $2\times 2$ polynomial matrix 
$(a_{ij}))$ such that 

$$
\left(
\begin{array}{cc}
 a_{11} &  a_{12}  \\
 a_{21} &  a_{22} 
\end{array}
\right)
\left(
\begin{array}{cc}
 g_{11} &  g_{12}  \\
 0 &  g_{22} 
\end{array}
\right) = 
\left(
\begin{array}{cc}
 x^m-1&  0  \\
 0 &  x^m-1
\end{array}
\right).
$$
Then $a_{21}=0$,  $a_{11}=\frac{x^m-1}{g_{11}(x)}$,  $a_{22}=\frac{x^m-1}{g_{22}(x)}$,  and 
$\frac{x^m-1}{g_{11}(x)} g_{12} + a_{12} g_{22} =0$, which implies  $g_{11}(x)g_{22}(x) \mid  (x^m-1) g_{12}(x)$. 

2) Consider the projection mapping $P : C \rightarrow R$, $P\big(a(x),b(x)\big) = a(x)$.  
Then $\ker P = \langle \big(0,g_{22}(x)\big) \rangle$  and ${\rm Im} \ P = \langle g_{11}(x) \rangle$. 
If $C$ is generated by elements $(0,g_{22}(x))$ and $\big( g_{11}(x),g'_{12}(x)\big)$, then $g_{12}(x) - g'_{12}(x)$ 
is divisible by   $g_{22}(x)$. 
\end{proof}

\begin{remark}
If $\gcd(q,m)=1$ then the condition $g_{11}(x)g_{22}(x) \mid  (x^m-1) g_{12}(x)$ is equivalent to the condition   
$\gcd \big( g_{11}(x),g_{22}(x)\big) \mid g_{12}(x)$, since $x^m-1$ has no multiple roots. 
\end{remark}

\begin{remark}
Let code $C$ be generated by $(g_{11}(x), g_{12}(x))$ and $(0,g_{22}(x))$, where 
$g_{11}(x)\mid (x^m-1)$   and  $g_{22}(x)\mid (x^m-1)$, and $\deg g_{12} (x) < \deg g_{22}(x)$.  
Several papers, for example \cite{Biswas}, make the incorrect statement that $\dim C = 2m-\deg g_{11}(x)-\deg g_{22}(x)$ in this case.  In order to have the equality $\dim C = 2m-\deg g_{11}(x)-\deg g_{22}(x)$, the condition 
$g_{11}(x)g_{22}(x) \mid  (x^m-1) g_{12}(x)$ is necessary, as it is shown in the next example. 
Let $x^m-1= p_1p_2p_3$ be a product of three irreducible polynomials (as in case $m=7$, $q=2$, for example). 
Let a code $C$ be generated by $(p_1p_2, 1)$ and $(0, p_3)$, 
and let a code $C'$ be generated by $(p_1p_2, 1)$ and $(0, p_1p_3)$. 
It is easy to see that $C=C'$, but we have  
$\dim C = 2m-\deg g_{11}(x)-\deg g_{22}(x)$ and   $\dim C' \ne 2m-\deg g_{11}(x)-\deg g_{22}(x)$.  
Theorem \ref{generators}  states that for any QC code, generators $(g_{11}(x), g_{12}(x))$ and $(0,g_{22}(x))$ 
{\bf can be chosen} such that they 
satisfy conditions (*) (in particular,  $\dim C = 2m-\deg g_{11}(x)-\deg g_{22}(x)$). 
In this case, any QC code is determined uniquely by their generators $(g_{11}(x), g_{12}(x))$ and $(0,g_{22}(x))$, when polynomials 
$g_{11}(x)$ and $g_{22}(x)$ are monic. 
It is a canonical representation of a code, and we use this representation in the remaining part of the paper. 
This representation corresponds to the reduced Gr\"{e}bner basis from  \cite{Lally}. 
\end{remark}

Note that a special case of quasi-cyclic codes with  $g_{11}(x) \mid g_{12}(x)$, satisfying Conditions $(\ast)$,  
was considered in \cite{Galindo}.

Now we  study the case when a quasi-cyclic code of index 2 can be generated by one element.

\begin{theorem}
\label{one-gen}
Let $\gcd(q,m)=1$. Let $C$ be a quasi-cyclic code  of length $2m$ and index $2$, generated by  elements 
$g_1=\big( g_{11}(x),g_{12}(x)\big)$ and $g_2=\big( 0,g_{22}(x)\big)$, satisfying Conditions $(\ast)$. Then 
$C$ is generated by one element if and only if $g_{11}(x) g_{22}(x) \equiv 0  \pmod{x^m-1}$. 
\end{theorem}

\begin{proof} 
Assume that $C$ is generated by $\big( g(x),g'(x)\big)$. Then there exist $a(x)$ and $b(x)$ such that  
$$a(x) \big( g(x),g'(x)\big)  \equiv  \big( g_{11}(x), g_{12}(x)\big)  \pmod{x^m-1},$$
$$b(x) \big( g(x),g'(x)\big)  \equiv  \big( 0, g_{22}(x)\big)  \pmod{x^m-1}.$$
Then 
$$b(x)g_{11}(x) \equiv b(x) a(x) g(x) \equiv b(x) g(x) a(x) \equiv 0 \pmod{x^m-1}.$$
Therefore, 
$$b(x) \equiv  0 \pmod{\frac{x^m-1}{g_{11}(x)}}$$
and 
$$g_{22}(x) \equiv b(x)g'(x) \equiv  0 \pmod{\frac{x^m-1}{g_{11}(x)}}.$$
Hence $g_{11}(x) g_{22}(x) \equiv 0  \pmod{x^m-1}$.

Conversely, suppose that   $g_{11}(x) g_{22}(x) \equiv 0  \pmod{x^m-1}$.  
We will show that  
$$A =\big( g_{11}(x),g_{12}(x) + b(x)g_{22}(x) \big)$$ 
generates the code $C$ for some $b(x)$. 

Let $d(x)= \gcd(g_{11}(x), g_{22}(x))$, $g_{11}(x) = g'_{11}(x) d(x)$, and $g_{22}(x) = g'_{22}(x) d(x)$. 
Then we have $\gcd(g'_{11}(x), g'_{22}(x))=1$, $(x^m-1)d(x) =g_{11}(x)g_{22}(x)$, $x^m-1 = d(x)g'_{11}(x)g'_{22}(x)$ 
and $ \frac{x^m-1}{g_{11}(x)} = g'_{22}(x)$. 
We also have $ \frac{x^m-1}{g_{11}(x)} g_{12}(x) = g_{22}(x) h(x)$ for some $h(x)$, 
since $g_{22}(x) \mid \frac{x^m-1}{g_{11}(x)} g_{12}(x)$. 

Consider $\frac{x^m-1}{g_{11}(x)} A = \big( 0, g_{22}(x) h(x) + g'_{22}(x) b(x)g_{22}(x) \big)$. 
If $g_{22}(x) h(x)R = g_{22}(x)R$, then we can choose $b(x)=0$, so $A=g_1$ and  $AR$ contains  $g_2$, hence   
$A$ generates $C$. 

Now assume that  $g_{22}(x) h(x)R \subsetneqq g_{22}(x)R$.  Since $\gcd(g'_{22}(x),(x^m-1)/g_{22}(x)) =1$, 
there exists $f(x)$, such that 
$$f(x)g'_{22}(x) \equiv 1 \pmod{\frac{x^m-1}{g_{22}(x)}},$$ 
which is equivalent to 
$$f(x)g'_{22}(x)g_{22}(x) \equiv g_{22}(x) \pmod{x^m-1}.$$

Let $g'_{11}(x) = a_1(x) \cdots a_t(x)$ be a product of irreducible polynomials. We define  
$$b(x) = f(x) \left( \sum_{i=1}^t \frac{g'_{11}(x)}{a_i(x)} -h(x) \right) .$$
Then 
\begin{eqnarray*}
g_{22}(x) h(x) + g'_{22}(x) g_{22}(x) b(x) 
&\equiv&  g_{22}(x) h(x) + g'_{22}(x) g_{22}(x)  f(x) \left( \sum_{i=1}^t  \frac{g'_{11}(x)}{a_i(x)} -h(x) \right) \\
&\equiv&  g_{22}(x) h(x) +  g_{22}(x)  \left( \sum_{i=1}^t  \frac{g'_{11}(x)}{a_i(x)} -h(x) \right) \\
&\equiv& g_{22}(x)  \sum_{i=1}^t  \frac{g'_{11}(x)}{a_i(x)} \pmod{x^m-1} .
\end{eqnarray*}

We note that 
$$g_{22}(x) \sum_{i=1}^t \frac{g'_{11}(x)}{a_i(x)}  \equiv  g_{22}(x) \frac{g'_{11}(x)}{a_s(x)} \pmod{g_{22}(x)a_s(x)} 
\not\equiv 0  \pmod{g_{22}(x)a_s(x)}$$
for all $1\le s \le t$. 
Since $x^m-1 = g_{22}(x) g'_{11}(x)$, we have  
$$\big( g_{22}(x) h(x) + g'_{22}(x) g_{22}(x) b(x) \big) R = g_{22}(x)R.$$
Hence, $AR$ contains  $g_2$ and $g_1= A- b(x)g_2$, thus $A$ generates $C$. 
\end{proof} 

In fact, one part of Theorem \ref{one-gen} is valid without restrictions on the characteristic of the ground field. 
The proof of Theorem \ref{one-gen} suggests  the following: 

\begin{proposition}
\label{one-gen-prop}
Let $C$ be a quasi-cyclic code  of length $2m$ and index $2$, generated by  elements 
$g_1=\big( g_{11}(x),g_{12}(x)\big)$ and $g_2=\big( 0,g_{22}(x)\big)$, satisfying Conditions $(\ast)$. If  
$C$ is generated by one element then  $g_{11}(x) g_{22}(x) \equiv 0  \pmod{x^m-1}$. 
\end{proposition}

Esmaeili and Yari \cite{Esmaeili} also characterized one-generator codes. According to their criterion, 
in order to find out whether a code $C$ is one-generator, 
one has to find the decomposition of the code $C$ into a sum of component codes and study their properties, 
while our Theorem \ref{one-gen} requires one to check only one condition, viz. if $g_{11}(x) g_{22}(x) \equiv 0  \pmod{x^m-1}$.

Now we estimate minimum  distances of quasi-cyclic codes. 
Let $C$ be a quasi-cyclic code  of length $2m$ and index 2, generated by elements 
$g_1=\big( g_{11}(x),g_{12}(x)\big)$ and $g_2=\big( 0,g_{22}(x)\big)$, satisfying Conditions $(\ast)$. 
Define the following cyclic codes: 
\begin{equation} 
\label{4codes}
\begin{array}{ccc}
C_1 &=& \langle g_{11}(x) \rangle, \\
C_2 &=& \langle g_{22}(x) \rangle, \\
C_3 &=& \big\langle \gcd (g_{12}(x), g_{22}(x)) \big\rangle, \\
C_4 &=& \left\langle \frac{g_{11}(x)g_{22}(x)}{\gcd (g_{12}(x), g_{22}(x)) } \right\rangle.
\end{array}
\end{equation}


Let $d(C)$ denote the minimum Hamming distance of  $C$. We adopt the convention that $d(\{ 0\}) = \infty$.  

\begin{theorem}
\label{distance}
Let $C$ be a quasi-cyclic code  of length $2m$ and index $2$, generated by  elements 
$g_1=\big( g_{11}(x),g_{12}(x)\big)$ and $g_2=\big( 0,g_{22}(x)\big)$, satisfying Conditions $(\ast)$. Then 
\begin{equation} 
\label{bound}
d(C) \ge \min \{ d(C_2), d(C_4), d(C_1) + d(C_3) \}. 
\end{equation} 
\end{theorem}

\begin{proof} 
Assume that a codeword $c \in C$ has the form $(0,B)$, $B\ne 0$. Then 
$$c = a(x) \frac{x^m-1}{g_{11}(x)} g_1 + b(x)g_2 =  \big(0, a(x)\frac{x^m-1}{g_{11}(x)} g_{12}(x) + b(x)g_{22}(x) \big).$$ 
Since $g_{22}(x) \mid \frac{x^m-1}{g_{11}(x)} g_{12}(x) $, we have 
$$wt (c) \ge d(C_2),$$ 
where $wt(c)$ denotes the Hamming weight of $c$. 

Assume that  $c \in C$ has the form $(A,0)$, $A\ne 0$. Then 
$$c = a(x) g_1 + b(x) g_2 =  \big( a(x)g_{11}(x), a(x) g_{12}(x) + b(x)g_{22}(x) \big),$$ 
where $a(x) g_{12}(x) + b(x)g_{22}(x)  \equiv 0 \pmod{x^m-1}$. 
If $b(x) \not\equiv 0 \pmod{x^m-1}$,  then $a(x) g_{12}(x) \in \langle g_{22}(x) \rangle =C_2$. Hence 
$$a(x)= \frac{g_{22}(x)}{\gcd(g_{12}(x),g_{22}(x))} r, \ \ r \in R.$$
Then 
$$a(x)g_{11}(x)= \frac{g_{22}(x)g_{11}(x)}{\gcd(g_{12}(x),g_{22}(x))}r \in C_4.$$
If $b(x) \equiv 0 \pmod{x^m-1}$ then 
$$a(x)= \frac{x^m-1}{\gcd(g_{12}(x), x^m-1)} r, \ \ r \in R.$$
Hence  
$$a(x)g_{11}(x)= \frac{(x^m-1)g_{11}(x)}{\gcd(g_{12}(x), x^m-1)}r \in C_4,$$
since $\frac{g_{22}(x)}{\gcd(g_{12}(x),g_{22}(x))}$ divides $\frac{x^m-1}{\gcd(g_{12}(x),x^m-1)}$. 
Therefore, 
$$wt (c) \ge d(C_4).$$

Finally, assume that  $c \in C$ has the form $(A,B)$, $A\ne 0$, $B\ne 0$. Then 
$$c =  \big( a(x)g_{11}(x), a(x) g_{12}(x) + b(x)g_{22}(x) \big).$$ 
Since $a(x) g_{12}(x) + b(x)g_{22}(x) \in C_3$, we have 
$$wt (c) \ge d(C_1) + d(C_3),$$
which proves the theorem. 
\end{proof} 

\begin{remark}
The lower bound (\ref{bound}) is sharp in the sense that there is a code $C$ for which the bound 
is exact. 
Indeed, consider $g_{11}(x)=g_{22}(x)$ and $g_{12}(x)=0$. Then $C_1=C_2=C_3=C_4$ and 
$C=C_1\oplus C_1$. 
As a result, we have equality in (\ref{bound}). 
\end{remark}

\begin{remark}
The lower bound (\ref{bound}) looks similar to the bound in \cite{Dastbasteh}, however the paper \cite{Dastbasteh} only  considered bounds for codes over the extension field 
${\mathbb F}_{p^2}$, not over ${\mathbb F}_p$. The bound  in \cite{Dastbasteh} is similar to the bound for symplectic weights in Theorem \ref{distance-s}. 
\end{remark}

\section{Euclidean duals of quasi-cyclic codes}
\label{euclidean}

Let $f(x)=a_0 +a_1 x+a_2 x+\cdots +a_k x^k$ be a polynomial of degree $k$. Then the {\emph{reciprocal polynomial}} 
of $a(x)$ is the polynomial 
$$f(x)^*=x^{\deg f} f(x^{-1}) =a_k +a_{k-1} x+\cdots +a_{0}x^k. $$


Let $f(x)=a_0 +a_1 x+ \cdots+ a_{m-1}x^{m-1}+a_{m}x^{m} $, where $m$ is as before, i.e., $R=F[x]/ \langle x^m-1\rangle$.   
We define the {\emph{transpose polynomial}} $f(x)^\circ$ of $f(x)$ as 
$$f(x)^\circ =x^m f(x^{-1})=a_m +a_{m-1} x+\cdots+a_{2}x^{m-2}+a_{1}x^{m-1} + a_0x^m. $$

\begin{lemma} 
\label{transpose} 
1) If $\deg f(x) \le m$ then 
$$f(x)^\circ =  x^{m-\deg f}  f(x)^* ,$$
$$f(x)^\circ \equiv (a_0 + a_m) +a_{m-1} x+\cdots+a_{2}x^{m-2}+a_{1}x^{m-1} \pmod{x^m-1}.$$ 
 
2) If $\deg f(x)h(x) \le m$ then 
$$x^m \big( f(x)h(x)\big)^\circ =  f(x)^{\circ} h(x)^{\circ},$$ 
$$\big( f(x)h(x)\big)^\circ \equiv  f(x)^{\circ} h(x)^{\circ} \pmod{x^m-1}.$$ 
\end{lemma}

\begin{proof}
We have
 $$f(x)^\circ = x^mf(x^{-1}) = x^{m-\deg f} x^{\deg f}f(x^{-1})  =  x^{m-\deg f}  f(x)^*,$$ 
$$x^m \big( f(x)h(x)\big)^\circ = x^m x^m f(x^{-1})h(x^{-1})  =   f(x)^{\circ} h(x)^{\circ} ,$$   
which proves the lemma.
 \end{proof}
 
We recall the standard inner product on the space $R=F[x]/ \langle x^m-1 \rangle$. 
Let $a(x)=a_0+a_1x+\cdots+a_{m-1}x^{m-1}$ and $b(x)=b_0+b_1x+\cdots+b_{m-1}x^{m-1}$. 
Then the Euclidean inner product on $R$ is defined as 
\begin{equation}
\label{innerprod}
\big\langle a(x),b(x)\big\rangle_e = a_0b_0+a_1b_1+ \cdots +a_{m-1}b_{m-1}. 
\end{equation}
It is consistent with the standard dot product between vectors $(a_0,a_1,\dots, a_{m-1})$ and $(b_0,b_1,\dots, b_{m-1})$. 

\begin{lemma}[\cite{Galindo}] 
\label{adjoint} 
Let $a(x)$, $b(x)$, $c(x)$ be polynomials in $R$. Then   
\[ \big\langle c(x)a(x),b(x)\big\rangle_e = \big\langle  a(x),c(x)^\circ b(x) \big\rangle_e .\]
\end{lemma}

\begin{proof}
It is sufficient to prove the statement for $c(x)=x^k$. We have to show that 
\[ \big\langle x^k a(x),b(x)\big\rangle_e = \big\langle  a(x), x^{m-k} b(x)  \big\rangle_e ,\]
which is equivalent to the equality 
\[ \sum_{i=0}^{m-1} a_{i-k}b_i =  \sum_{i=0}^{m-1} a_ib_{i+k} ,\]
where indices are considered modulo $m$. 
\end{proof}

\begin{remark}
Lemma \ref{adjoint} explains why $c(x)^{\circ}$ is called the transpose polynomial for $c(x)$. 
Recall that if $\varphi$ is an endomorphism of a vector space $R$, then the adjoint 
$\varphi^{\circ}$ of $\varphi$ is defined by the equation 
$\big\langle \varphi(a),b\big\rangle_e = \big\langle  a,\varphi^{\circ}(b) \big\rangle_e$. 
For symmetric inner products, in the matrix presentation, the adjoint of a matrix is the transpose of the matrix. 
If we consider the multiplication by $c(x)$ as an endomorphism of $R$, then its adjoint  is the 
multiplication by $c(x)^\circ$.  
\end{remark}

\begin{lemma} 
\label{gg} 
For $f(x) \in  F[x]$, the following statements hold: 

1) If $f(x)$ divides $x^m-1$, then $f(x)^{\circ}$ divides $x^m(1-x^m)$. 

2) If $f(x)$ divides $x^m-1$, then 
$$\left( \frac{x^m-1}{f(x)} \right)^{\circ} = \frac{x^m(1-x^m)}{f(x)^{\circ}}.$$ 

3) If $\deg g_{12} (x) < \deg g_{22}(x)$ and $g_{11}(x) g_{22}(x)$ divides $(x^m-1)g_{12}(x)$, then 
$$\left( \frac{(x^m-1)g_{12}(x)}{g_{11}(x)g_{22}(x)} \right)^{\circ} = 
\frac{x^m(1-x^m) g_{12}(x)^{\circ}}{g_{11}(x)^{\circ}g_{22}(x)^{\circ}} .$$
\end{lemma}

\begin{proof} 
We have 
$$\left( \frac{x^m-1}{f(x)} \right)^{\circ}  =  x^m \frac{(1/x)^m-1}{f(x^{-1})} =  x^m \frac{1-x^m}{x^mf(x^{-1})} 
 =  \frac{x^m(1-x^m)}{f(x)^{\circ}} .$$
 
Let $h(x) \cdot g_{11}(x)g_{22}(x) = (x^m-1)g_{12}(x)$.  Then $\deg h(x) < m$ and 
$$x^mh(1/x) \cdot x^m g_{11}(1/x) x^m g_{22}(1/x) = x^mx^m((1/x)^m-1) x^mg_{12}(1/x),$$
$$h(x)^{\circ} \cdot g_{11}(x)^{\circ}  g_{22}(x)^{\circ}  = x^m (1-x^m) g_{12}(x)^{\circ}.$$
Thus the lemma follows.
\end{proof} 

\begin{remark}
If $f(x) \mid (x^m-1)$, then $f(x)^* \mid (x^m-1)$. However,   $f(x)^{\circ}$ might not divide $x^m-1$, 
but it divides $x^m(x^m-1)$. 
\end{remark}

The inner product (\ref{innerprod}) can be naturally extended to quasi-cyclic codes of length $n=2m$ and index 2: 
\[ \big\langle  \big(a(x),b(x)\big) , \big(a'(x),b'(x)\big) \big\rangle_e  =  
\big\langle  a(x),a'(x) \big\rangle_e  + \big\langle  b(x), b'(x) \big\rangle_e.    \]

If $C$ is a code in $F^{n}$, then its Euclidean dual code is 
$$C^{\perp_e} = \{ u \in F^{n} \mid \langle u, v \rangle =0, \ {\rm for \ all} \ v\in C  \}.$$
The code $C$ is called  self-orthogonal if $C \subseteq C^{\perp_e}$, 
and  dual-containing if $C \supseteq C^{\perp_e}$.

\begin{theorem}
\label{dual}
Let $C$ be a quasi-cyclic code  of length $2m$ and index $2$, generated by two elements 
$g_1=\big( g_{11}(x),g_{12}(x)\big)$ and $g_2=\big( 0,g_{22}(x)\big)$, satisfying Conditions $(\ast)$. 
Then its Euclidean dual code $C^{\perp_e}$ is generated by 
two elements $\Big(  \frac{x^m(x^m-1)}{g_{11}(x)^{\circ}} ,0 \Big)$ and 
$\Big( \frac{x^m(x^m-1)g_{12}(x)^{\circ}} {g_{11}(x)^{\circ}  g_{22}(x)^{\circ} } ,   
- \frac{x^m(x^m-1)}{g_{22}(x)^{\circ} } \Big)$. 
\end{theorem}

\begin{proof} 
Let code the $C'$ be generated by two elements $f_1=\Big(  \frac{x^m(x^m-1)}{g_{11}(x)^{\circ}} ,0 \Big)$ and \\
$f_2=\Big( \frac{x^m(x^m-1)g_{12}(x)^{\circ}} {g_{11}(x)^{\circ}  g_{22}(x)^{\circ} } ,   
- \frac{x^m(x^m-1)}{g_{22}(x)^{\circ} } \Big)$. Then 
$$\langle g_2, f_1 \rangle_e = 0,$$
$$\langle g_1, f_1 \rangle_e =
\langle g_{11}(x), \frac{x^m(x^m-1)}{g_{11}(x)^{\circ}}  \rangle_e  = 
\langle 1, \frac{x^m(x^m-1)}{g_{11}(x)^{\circ}}  g_{11}(x)^{\circ} \rangle_e  = 0,$$
$$\langle g_2, f_2 \rangle_e =
\langle g_{22}(x), - \frac{x^m(x^m-1)}{g_{22}(x)^{\circ} } \rangle_e  = 
\langle 1, -\frac{x^m(x^m-1)}{g_{22}(x)^{\circ}}  g_{22}(x)^{\circ} \rangle_e  = 0,$$
\begin{eqnarray*}
\langle g_1, f_2 \rangle_e 
&=& \langle g_{11}(x), \frac{x^m(x^m-1)g_{12}(x)^{\circ}} {g_{11}(x)^{\circ}  g_{22}(x)^{\circ} } \rangle_e + 
\langle g_{12}(x), - \frac{x^m(x^m-1)}{g_{22}(x)^{\circ} } \rangle_e  \\
&=& \langle 1, \frac{x^m(x^m-1)g_{12}(x)^{\circ}} {g_{22}(x)^{\circ}  } \rangle_e + 
\langle 1, - \frac{x^m(x^m-1) g_{12}(x)^{\circ}}{g_{22}(x)^{\circ} } \rangle_e  \\
&=& 0 .
\end{eqnarray*}
Thus $C' \subseteq C^{\perp_e}$. We will show now that $\dim C' = \dim C^{\perp_e}$. 
Note that $\dim C^{\perp_e} = 2m - \dim C = \deg g_{11}(x) + \deg g_{2}(x)$. 

Using Lemma \ref{transpose}, we can rewrite generators  $f_1$ and $f_2$ in terms of reciprocal polynomials: 
$$f_1 = x^{\deg g_{11}}\Big(  \frac{x^m-1}{g_{11}(x)^{*}} ,0 \Big), $$
 \begin{eqnarray*}
f_2 
& \equiv & \Big( \frac{x^m(x^m-1)g_{12}(x)^{\circ}} {g_{11}(x)^{\circ}  g_{22}(x)^{\circ} }x^m ,   
- \frac{x^m(x^m-1)}{g_{22}(x)^{\circ} } \Big)  \pmod{x^m-1}\\
& \equiv & x^{\deg g_{22}}\Big( \frac{x^{m+\deg g_{11} - \deg g_{12}} (x^m-1)g_{12}(x)^{*}} {g_{11}(x)^{*}  g_{22}(x)^{*}} ,   
- \frac{x^m-1}{g_{22}(x)^{*} } \Big)  \pmod{x^m-1}.
\end{eqnarray*}

By Theorem \ref{generators}, 
$$\dim C' = 2m - \deg \frac{x^m-1}{g_{11}(x)^{*}}  - \deg \frac{x^m-1}{g_{22}(x)^{*}} = \deg g_{11}(x) + \deg g_{2}(x),$$
which completes the proof. 
\end{proof}

The proof of  Theorem  \ref{dual} shows that it can  be reformulated as 
\begin{corollary}
\label{dual-cor}
Let $C$ be a quasi-cyclic code  of length $2m$ and index $2$, generated by two elements 
$g_1=\big( g_{11}(x),g_{12}(x)\big)$ and $g_2=\big( 0,g_{22}(x)\big)$, satisfying Conditions $(\ast)$. 
Then its dual code $C^{\perp_e}$ is generated by 
two elements $\Big(  \frac{x^m-1}{g_{11}(x)^{*}} ,0 \Big)$ and 
$\Big( \frac{x^{m+\deg g_{11} - \deg g_{12}}(x^m-1)g_{12}(x)^{*}} {g_{11}(x)^{*}  g_{22}(x)^{*}} ,   
- \frac{x^m-1}{g_{22}(x)^{*} } \Big)$. 
\end{corollary}

 Now we determine when $C^{\perp_e}$ is generated by one element 
 (this result is dual to Theorem \ref{one-gen} in some sense).
 
 \begin{theorem}
\label{dual-one-gen}
Let $\gcd(q,m)=1$. Let $C$ be a quasi-cyclic code  of length $2m$ and index 2, generated by  elements 
$g_1=\big( g_{11}(x),g_{12}(x)\big)$ and $g_2=\big( 0,g_{22}(x)\big)$, satisfying Conditions $(\ast)$. Then 
$C^{\perp_e}$ is generated by one element if and only if $g_{11}(x) g_{22}(x)$ divides $x^m-1$. 
\end{theorem}

\begin{proof} 
By Corollary \ref{dual-cor} and Theorem \ref{one-gen},    
$C^{\perp_e}$ is generated by one element if and only if 
$ \frac{x^m-1}{g_{11}(x)^{*}}  \cdot  \frac{x^m-1}{g_{22}(x)^{*} } \equiv 0  \pmod{x^m-1}$, 
if and only if $g_{11}(x)^{*} g_{22}(x)^{*} $ divides $x^m-1$, 
which is equivalent to the fact that $g_{11}(x) g_{22}(x)$ divides $x^m-1$. 
\end{proof} 

Now we investigate conditions for $C$ to be Euclidean self-orthogonal. 

\begin{theorem}
\label{self-orth} 
Let $C$ be a quasi-cyclic code  of length $2m$ and index $2$, generated by elements 
$g_1=\big( g_{11}(x),g_{12}(x)\big)$ and $g_2=\big( 0,g_{22}(x)\big)$, satisfying Conditions $(\ast)$. 
Then $C$ is self-orthogonal if and only if the following conditions hold:

1) $g_{22}(x) g_{22}(x)^{\circ} \equiv 0  \pmod{x^m-1}$;

2) $g_{12}(x)g_{22}(x)^{\circ} \equiv 0  \pmod{x^m-1}$;

3) $g_{11}(x)g_{11}(x)^{\circ}  + g_{12}(x)g_{12}(x)^{\circ} \equiv 0  \pmod{x^m-1}$.
\end{theorem}

\begin{proof} 
Self-orthogonality of $C$ means that 
\[ \big\langle g_1, b(x) g_1 \big\rangle_e  = 
\big\langle  \big( g_{11}(x), g_{12}(x)\big) , b(x) \big( g_{11}(x), g_{12}(x)) \big\rangle_e  = 0,    \]
\[ \big\langle g_1, b(x) g_2 \big\rangle_e  = 
\big\langle  \big( g_{11}(x), g_{12}(x)\big) , b(x) \big( 0, g_{22}(x)) \big\rangle_e  = 0,    \]
\[ \big\langle g_2, b(x) g_2 \big\rangle_e  = \big\langle  \big( 0, g_{22}(x)\big) , b(x) \big( 0, g_{22}(x)) \big\rangle_e  = 0,    \]
for any $b(x) \in R$. This is equivalent to 
\[ \big\langle   g_{11}(x) g_{11}(x)^{\circ} +  g_{12}(x) g_{12}(x)^{\circ} , b(x) \big\rangle_e  = 0,    \]
\[ \big\langle   g_{12}(x) g_{22}(x)^{\circ} , b(x) \big\rangle_e  = 0,    \] 
\[ \big\langle   g_{22}(x) g_{22}(x)^{\circ} , b(x) \big\rangle_e  = 0,    \]
for any $b(x) \in R$, which proves the theorem.  
\end{proof}

Similarly, we can give a description of dual-containing codes. 

\begin{theorem}
\label{dual-con} 
Let $C$ be a quasi-cyclic code  of length $2m$ and index $2$, generated by elements 
$g_1=\big( g_{11}(x),g_{12}(x)\big)$ and $g_2=\big( 0,g_{22}(x)\big)$, satisfying Conditions $(\ast)$. 
Then $C$ is dual-containing if and only if the following conditions hold:

1) $g_{11}(x) g_{11}(x)^{\circ}$ divides  $x^m(x^m-1)$;

2) $g_{11}(x) g_{11}(x)^{\circ} g_{22}(x)$ divides  $x^m(x^m-1)g_{12}(x)$ ;

3) $g_{11}(x) g_{11}(x)^{\circ} g_{22}(x) g_{22}(x)^{\circ}$ divides  
$x^m(x^m-1)(g_{11}(x) g_{11}(x)^{\circ} + g_{12}(x)g_{12}(x)^{\circ}) $. 

\end{theorem}

\begin{proof} 
$C$ is dual-containing if and only if $C^{\perp_e}$ is self-orthogonal. By Theorem \ref{dual} it means that 
\[ \big\langle  \big( \frac{x^m(x^m-1)}{g_{11}(x)^{\circ}}, 0 \big) ,  \big( \frac{x^m(x^m-1)}{g_{11}(x)^{\circ}} b(x), 0\big) \big\rangle_e  = 0,  \]
\[  \big\langle \big( \frac{x^m(x^m-1)}{g_{11}(x)^{\circ}} , 0\big) ,  \big( \frac{x^m(x^m-1) g_{12}(x)^{\circ}}{g_{11}(x)^{\circ} g_{22}(x)^{\circ}}b(x), -\frac{x^m(x^m-1)}{ g_{22}(x)^{\circ}} b(x) \big) \big\rangle_e  = 0,  \]
\[ \big\langle  \big( \frac{x^m(x^m-1) g_{12}(x)^{\circ}}{g_{11}(x)^{\circ} g_{22}(x)^{\circ}}, -\frac{x^m(x^m-1)}{ g_{22}(x)^{\circ}} \big) ,  
\big( \frac{x^m(x^m-1) g_{12}(x)^{\circ}}{g_{11}(x)^{\circ} g_{22}(x)^{\circ}} b(x), -\frac{x^m(x^m-1)}{ g_{22}(x)^{\circ}} b(x) \big)\big\rangle_e  = 0,  \]
for any $b(x) \in R$. This is equivalent to  
\[  \frac{x^m(x^m-1)}{g_{11}(x)^{\circ}}      \left( \frac{x^m(x^m-1)}{g_{11}(x)^{\circ}}  \right) ^{\circ}  \equiv 0  \pmod{x^m-1},  \] 
\[  \frac{x^m(x^m-1)}{g_{11}(x)^{\circ}}    \left(  \frac{x^m(x^m-1)g_{12}(x)^{\circ}}{g_{11}(x)^{\circ}g_{22}(x)^{\circ}}  \right) ^{\circ}  \equiv 0  \pmod{x^m-1},  \]
\[  \frac{x^m(x^m-1)g_{12}(x)^{\circ}}{g_{11}(x)^{\circ}g_{22}(x)^{\circ}}   \left( \frac{x^m(x^m-1)g_{12}(x)^{\circ}}{g_{11}(x)^{\circ}g_{22}(x)^{\circ}}  \right) ^{\circ}  + 
 \frac{x^m(x^m-1)}{g_{22}(x)^{\circ}}  \left( \frac{x^m(x^m-1)}{g_{22}(x)^{\circ}}  \right) ^{\circ}  \equiv 0  \pmod{x^m-1}, \]
i.e.,  
\[  \frac{x^m(x^m-1)}{g_{11}(x)^{\circ}}  \cdot \frac{(x^m-1)}{g_{11}(x)} \equiv 0  \pmod{x^m-1},  \] 
\[  \frac{x^m(x^m-1)}{g_{11}(x)^{\circ}}    \cdot \frac{(x^m-1)g_{12}(x)}{g_{11}(x)g_{22}(x)} \equiv 0  \pmod{x^m-1},  \]
\[  \frac{x^m(x^m-1)g_{12}(x)^{\circ}}{g_{11}(x)^{\circ}g_{22}(x)^{\circ}}   \cdot \frac{(x^m-1)g_{12}(x)}{g_{11}(x)g_{22}(x)}    + 
 \frac{x^m(x^m-1)}{g_{22}(x)^{\circ}}  \cdot \frac{(x^m-1)}{g_{22}(x)}  \equiv 0  \pmod{x^m-1}, \]
 which proves the theorem.  
\end{proof}

In terms of reciprocal polynomials the previous theorems can be reformulated as follows. 

\begin{corollary}
\label{self-orth-cor} 
Let $C$ be a quasi-cyclic code  of length $2m$ and index $2$, generated by elements 
$g_1=\big( g_{11}(x),g_{12}(x)\big)$ and $g_2=\big( 0,g_{22}(x)\big)$, satisfying Conditions $(\ast)$. 
Then $C$ is self-orthogonal if and only if the following conditions hold:

1) $g_{22}(x) g_{22}(x)^* \equiv 0  \pmod{x^m-1}$;

2) $g_{12}(x)g_{22}(x)^* \equiv 0  \pmod{x^m-1}$;

3) $x^{\deg g_{12}} g_{11}(x)g_{11}(x)^*  + x^{\deg g_{11}}  g_{12}(x)g_{12}(x)^* \equiv 0  \pmod{x^m-1}$.
\end{corollary}

\begin{corollary}
\label{dual-con-cor} 
Let $C$ be a quasi-cyclic code  of length $2m$ and index $2$, generated by elements 
$g_1=\big( g_{11}(x),g_{12}(x)\big)$ and $g_2=\big( 0,g_{22}(x)\big)$, satisfying Conditions $(\ast)$. 
Then $C$ is dual-containing if and only if the following conditions hold:

1) $g_{11}(x) g_{11}(x)^*$ divides  $x^m-1$;

2) $g_{11}(x) g_{11}(x)^* g_{22}(x)$ divides  $(x^m-1)g_{12}(x)$ ;

3) $g_{11}(x) g_{11}(x)^* g_{22}(x) g_{22}(x)^*$ divides  
$(x^m-1)(x^{\deg g_{12}} g_{11}(x)g_{11}(x)^*  + x^{\deg g_{11}}  g_{12}(x)g_{12}(x)^*) $. 
\end{corollary}


\section{Quasi-cyclic codes with symplectic inner product}
\label{symplectic}

Let $C$ be a quasi-cyclic code of length $2m$ and index 2. We define a symplectic form on $C\times C$ as 
\[ \big\langle  \big(a(x),b(x)\big) , \big(a'(x),b'(x)\big) \big\rangle _s =  
\big\langle  a(x),b'(x) \big\rangle_e  - \big\langle  b(x), a'(x) \big\rangle_e  .  \]

If $C$ is a code in $F^{2m}$, then its symplectic dual is 
$$C^{{\perp}_s} = \{ u \in F^{2m} \mid \langle u, v \rangle_s =0, \ {\rm for \ all} \ v\in C  \}.$$
The code $C$ is called symplectic self-orthogonal if $C \subseteq C^{{\perp}_s}$, 
and symplectic dual-containing if $C \supseteq C^{{\perp}_s}$. 

The next theorem can be proved in a manner similar to  Theorem \ref{dual}. 
 
\begin{theorem}
\label{dual-s}
Let $C$ be a quasi-cyclic code  of length $2m$ and index $2$, generated by two elements 
$g_1=\big( g_{11}(x),g_{12}(x)\big)$ and $g_2=\big( 0,g_{22}(x)\big)$, satisfying Conditions $(\ast)$. 
Then its symplectic dual $C^{{\perp}_s}$  is generated by two elements 
$\Big( 0, \frac{x^m(x^m-1)}{g_{11}(x)^{\circ}}  \Big)$ and 
$\Big( -\frac{x^m(x^m-1)}{g_{22}(x)^{\circ} } ,  
\frac{x^m(x^m-1)g_{12}(x)^{\circ}} {g_{11}(x)^{\circ}  g_{22}(x)^{\circ} }   \Big)$. 
\end{theorem}

Dastbasteh and Shivji \cite{Dastbasteh} described the symplectic dual of an additive cyclic code $C$ over 
${\mathbb F}_{p^2}$, assuming that $\gcd(m,p)=1$ and the decomposition of the code $C$ into  components is known. 
In Theorem \ref{dual-s}, only generators of a code were used to describe the symplectic dual code, 
and there is no restrictions on the field characteristic.

In terms of reciprocal polynomials the previous theorem can be reformulated as follows. 

\begin{corollary}
\label{dual-s-cor}
Let $C$ be a quasi-cyclic code  of length $2m$ and index $2$, generated by two elements 
$g_1=\big( g_{11}(x),g_{12}(x)\big)$ and $g_2=\big( 0,g_{22}(x)\big)$, satisfying Conditions $(\ast)$. 
Then its symplectic dual $C^{{\perp}_s}$  is generated by two elements 
$ \Big( 0, \frac{x^m-1}{g_{11}(x)^{*}}  \Big)$ and 
$\Big( - \frac{x^m-1}{g_{22}(x)^{*} } , 
\frac{x^{m+\deg g_{11} - \deg g_{12}}(x^m-1)g_{12}(x)^{*}} {g_{11}(x)^{*}  g_{22}(x)^{*}} \Big)$. 
\end{corollary}

Now we describe symplectic self-orthogonal codes. 

\begin{theorem}
\label{self-orth-s} 
Let $C$ be a quasi-cyclic code  of length $2m$ and index $2$, generated by elements 
$g_1=\big( g_{11}(x),g_{12}(x)\big)$ and $g_2=\big( 0,g_{22}(x)\big)$, satisfying Conditions $(\ast)$. 
Then $C$ is symplectic self-orthogonal if and only if the following conditions hold:

1) $g_{11}(x) g_{22}(x)^{\circ} \equiv 0  \pmod{x^m-1}$;

2) $g_{11}(x)g_{12}(x)^{\circ} - g_{12}(x)g_{11}(x)^{\circ} \equiv 0  \pmod{x^m-1}$.
\end{theorem}

\begin{proof} 
Symplectic self-orthogonality of $C$ means that 
\[ \big\langle  \big( g_{11}(x), g_{12}(x)\big) , b(x) \big( g_{11}(x), g_{12}(x)) \big\rangle_s  = 0,    \]
\[ \big\langle  \big( g_{11}(x), g_{12}(x)\big) , b(x) \big( 0, g_{22}(x)) \big\rangle_s  = 0,    \]
\[ \big\langle  \big( 0, g_{22}(x)\big) , b(x) \big( 0, g_{22}(x)) \big\rangle_s  = 0,    \]
for any $b(x) \in R$. This is equivalent to  
\[ \big\langle   g_{11}(x) g_{12}(x)^{\circ} -  g_{12}(x) g_{11}(x)^{\circ} , b(x) \big\rangle_e  = 0,    \]
\[ \big\langle   g_{11}(x) g_{22}(x)^{\circ} , b(x) \big\rangle_e  = 0,    \] 
for any $b(x) \in R$, which proves the theorem.  
\end{proof}

\begin{corollary}
\label{self-dual-s-cor}
Let $C$ be a quasi-cyclic code  of length $2m$ and index $2$, generated by one element 
$\big( 1,g_{12}(x)\big)$ with $g_{12}= g_{12}(x)^{\circ}$. 
Then $C$ is a  symplectic self-dual code. 
\end{corollary}

\begin{example} 
\label{self-dual-s}
Let $m=23$,  $g_{12}(x)= x^{22} + x^{18} + x^{14} + x^{11} + x^9 + x^5 + x$. 
Then $(1, g_{12}(x))$ generates a symplectic self-dual code  with minimum symplectic weight $8$. 
\end{example}

\begin{theorem}
\label{dual-con-s} 
Let $C$ be a quasi-cyclic code  of length $2m$ and index $2$, generated by elements 
$g_1=\big( g_{11}(x),g_{12}(x)\big)$ and $g_2=\big( 0,g_{22}(x)\big)$, satisfying Conditions $(\ast)$. 
Then $C$ is symplectic dual-containing if and only if the following conditions hold:

1) $g_{11}(x) g_{22}(x)^{\circ}$ divides  $x^m(x^m-1)$;

2) $g_{11}(x) g_{11}(x)^{\circ} g_{22}(x) g_{22}(x)^{\circ}$ divides  
$x^m(x^m-1) (g_{11}(x) g_{12}(x)^{\circ}-g_{11}(x)^{\circ} g_{12}(x))$.

\end{theorem}

\begin{proof} 
$C$ is symplectic dual-containing if and only if $C^{{\perp}_s}$ is symplectic self-orthogonal. By Theorem \ref{dual-s} it means that 
\[ \big\langle  \big( -\frac{x^m(x^m-1)}{g_{22}(x)^{\circ}}, \frac{x^m(x^m-1)g_{12}(x)^{\circ}}{g_{11}(x)^{\circ} g_{22}(x)^{\circ}} \big) ,  
\big( 0, \frac{x^m(x^m-1)}{g_{11}(x)^{\circ}} b(x) \big) \big\rangle_s  = 0,  \]
\[ \big\langle  \big( -\frac{x^m(x^m-1)}{g_{22}(x)^{\circ}}, \frac{x^m(x^m-1)g_{12}(x)^{\circ}}{ g_{11}(x)^{\circ} g_{22}(x)^{\circ}} \big) ,  
\big( -\frac{x^m(x^m-1)}{g_{22}(x)^{\circ}} b(x), \frac{x^m(x^m-1)g_{12}(x)^{\circ}}{ g_{11}(x)^{\circ} g_{22}(x)^{\circ}} b(x) \big)\big\rangle_s  = 0,  \]
for any $b(x) \in R$. This is equivalent to  
\[  \frac{x^m(x^m-1)}{g_{22}(x)^{\circ}}      \left( \frac{x^m(x^m-1)}{g_{11}(x)^{\circ}} \right) ^{\circ}  
\equiv 0  \pmod{x^m-1},  \] 
\[  \frac{x^m(x^m-1)}{g_{22}(x)^{\circ}}      \left( \frac{x^m(x^m-1)g_{12}(x)^{\circ}}{g_{11}(x)^{\circ} g_{22}(x)^{\circ}}  \right) ^{\circ}   -  
 \frac{x^m(x^m-1)g_{12}(x)^{\circ}}{g_{11}(x)^{\circ} g_{22}(x)^{\circ}}  \left( \frac{x^m(x^m-1)}{g_{22}(x)^{\circ}}  \right) ^{\circ}  \equiv 0  \pmod{x^m-1}, \]
i.e.,  
\[  \frac{x^m(x^m-1)}{g_{22}(x)^{\circ}}  \cdot \frac{(x^m-1)}{g_{11}(x)} \equiv 0  \pmod{x^m-1},  \] 
\[  \frac{x^m(x^m-1)}{g_{22}(x)^{\circ}}   \cdot \frac{(x^m-1)g_{12}(x)}{g_{11}(x)g_{22}(x)}    - 
 \frac{x^m(x^m-1)g_{12}(x)^{\circ}}{g_{11}(x)^{\circ}g_{22}(x)^{\circ}}  \cdot \frac{(x^m-1)}{g_{22}(x)}  \equiv 0  \pmod{x^m-1}, \]
and 
\[  \frac{x^m(x^m-1)}{g_{11}(x) g_{22}(x)^{\circ}}  \cdot (x^m-1) \equiv 0  \pmod{x^m-1},  \] 
\[  \frac{x^m(x^m-1) (g_{11}(x)g_{12}(x)^{\circ} - g_{11}(x)^{\circ}g_{12}(x))}{g_{11}(x)g_{11}(x)^{\circ}g_{22}(x)g_{22}(x)^{\circ}}   \cdot (x^m-1)   
  \equiv 0  \pmod{x^m-1}, \]
 which proves the theorem.  
\end{proof}

In terms of reciprocal polynomials the previous theorems can be reformulated as follows. 

\begin{corollary}
\label{self-orth-s-cor} 
Let $C$ be a quasi-cyclic code  of length $2m$ and index $2$, generated by elements 
$g_1=\big( g_{11}(x),g_{12}(x)\big)$ and $g_2=\big( 0,g_{22}(x)\big)$, satisfying Conditions $(\ast)$. 
Then $C$ is symplectic self-orthogonal if and only if the following conditions hold:

1) $g_{11}(x) g_{22}(x)^* \equiv 0  \pmod{x^m-1}$;

2) $x^{\deg g_{11}} g_{11}(x)g_{12}(x)^* - x^{\deg g_{12}} g_{12}(x)g_{11}(x)^* \equiv 0  \pmod{x^m-1}$.
\end{corollary}

\begin{corollary}
\label{dual-con-s-cor} 
Let $C$ be a quasi-cyclic code  of length $2m$ and index $2$, generated by elements 
$g_1=\big( g_{11}(x),g_{12}(x)\big)$ and $g_2=\big( 0,g_{22}(x)\big)$, satisfying Conditions $(\ast)$. 
Then $C$ is symplectic dual-containing if and only if the following conditions hold:

1) $g_{11}(x) g_{22}(x)^*$ divides  $x^m-1$;

2) $g_{11}(x) g_{11}(x)^* g_{22}(x) g_{22}(x)^*$ divides  
$(x^m-1) (x^{\deg g_{11}} g_{11}(x)g_{12}(x)^* - x^{\deg g_{12}} g_{12}(x)g_{11}(x)^*)$.
\end{corollary}

Corollaries \ref{self-orth-s-cor}  and \ref{dual-con-s-cor} imply a description of symplectic self-dual codes.

\begin{corollary}
\label{self-dual-s-cor2} 
Let $C$ be a quasi-cyclic code  of length $2m$ and index $2$, generated by elements 
$g_1=\big( g_{11}(x),g_{12}(x)\big)$ and $g_2=\big( 0,g_{22}(x)\big)$, satisfying Conditions $(\ast)$. 
Then $C$ is symplectic self-dual if and only if the following conditions hold:

1) $g_{11}(x) g_{22}(x)^* = \alpha(x^m-1)$ for some $\alpha \in F^*$; 

2) $x^{\deg g_{11}} g_{11}(x)g_{12}(x)^* - x^{\deg g_{12}} g_{12}(x)g_{11}(x)^*  = \beta(x^m-1)$ for some $\beta \in F^*$.
\end{corollary}

Recall that for a codeword $c= (a_1, \dots, a_n, b_1, \dots, b_n) \in C$ its symplectic weight is defined as 
$$wt_s(c) = | \{ i \mid (a_i,b_i)\ne (0,0) \} |.$$ 
The symplectic distance between codewords $c$, $e\in C$ is 
$$d_s(c,e) = wt_s(c-e),$$  
and the (minimum) symplectic distance $d_s(C)$ of a code $C$ is 
the smallest symplectic distance between distinct codewords of $C$.

Similarly as Theorem \ref{distance}, one can prove the following 

\begin{theorem}
\label{distance-s}
Let $C$ be a quasi-cyclic code  of length $2m$ and index $2$, generated by elements 
$g_1=\big( g_{11}(x),g_{12}(x)\big)$ and $g_2=\big( 0,g_{22}(x)\big)$, satisfying Conditions $(\ast)$. 
Let codes $C_1$, $C_2$, $C_3$, $C_4$ be defined as in (\ref{4codes}).  
Then 
$$d_s(C) \ge \min \{ d(C_2), d(C_4), \max \{d(C_1) , d(C_3)\} \}.$$ 
\end{theorem}

\section{Hermitian duals of quasi-cyclic codes}
\label{hermitian}

In this section we consider codes over the field $\mathbb{F}_{q^2}$ of $q^2$ elements. 
Let $a(x)=a_0+a_1x+\cdots+a_{m-1}x^{m-1}$ and $b(x)=b_0+b_1x+\cdots+b_{m-1}x^{m-1}$, where 
$a_i$, $b_i \in \mathbb{F}_{q^2}$ for any $0\le i \le m-1$. For $c\in \mathbb{F}_{q^2}$ we define $\overline{c}=c^q$ and 
$$\overline{a}(x) = \overline{a_0}+\overline{a_1}x+\cdots+\overline{a_{m-1}}x^{m-1}.$$
Then the Hermitian  inner product on $R$  is defined as 
\begin{equation}
\label{innerprod-h}
\big\langle a(x),b(x)\big\rangle _h = \big\langle a(x),\bar{b}(x)\big\rangle_e  = a_0\overline{b_0}+a_1\overline{b_1}+ \cdots +a_{m-1}\overline{b_{m-1}}. 
\end{equation}
If $C$ is a code in $(\mathbb{F}_{q^2})^{2m}$, then its Hermitian dual is 
$$C^{{\perp}_h} = \{ u \in (\mathbb{F}_{q^2})^{2m} \mid \langle u, v \rangle_h =0, \ {\rm for \ all} \ v\in C  \}.$$
The code $C$ is called Hermitian self-orthogonal if $C \subseteq C^{{\perp}_h}$, 
and Hermitian dual-containing if $C \supseteq C^{{\perp}_h}$.

Let $f(x)=a_0 +a_1 x+ \cdots+ a_{m-1}x^{m-1}+a_{m}x^{m} $.   
We define the {\emph{conjugate transpose polynomial}} $f(x)^{\dagger}$ of $f(x)$ as 
$$f(x)^\dagger =x^m \overline{f}(x^{-1})= \overline{a_{m}}  +\overline{a_{m-1}} x+\cdots+
\overline{a_{2}}x^{m-2}+\overline{a_{1}}x^{m-1} + \overline{a_0}x^m. $$

\begin{lemma} 
\label{transpose-h} 
If $\deg f(x) \le m$,  $\deg h(x) \le m$ and $\deg f(x)h(x) \le m$, then 

1) $f(x)^\dagger =  x^{m-\deg f}  \overline{f}^*(x) $;
 
2) $x^m \big( f(x)h(x)\big)^\dagger =  f(x)^{\dagger} h(x)^{\dagger}$; 

 3) $f(x)^\dagger \equiv (\overline{a_0} + \overline{a_m}) +\overline{a_{m-1}} x+\cdots+\overline{a_{2}}x^{m-2}+
 \overline{a_{1}}x^{m-1} \pmod{x^m-1}$; 

4) $\big( f(x)h(x)\big)^\dagger \equiv  f(x)^{\dagger} h(x)^{\dagger} \pmod{x^m-1}$. 
\end{lemma}

\begin{proof}
Similar to the proof of Lemma \ref{transpose}.  
 \end{proof}

The next lemma is an analog of Lemma \ref{adjoint} and it can be proved in the same way.
\begin{lemma}
\label{adjoint-h} 
Let $a(x)$, $b(x)$, $c(x)$ be polynomials in $R$. Then   
\[ \big\langle a(x)c(x),b(x)\big\rangle _h= \big\langle  a(x),b(x)c(x)^\dagger \big\rangle _h.\]
\end{lemma}

The inner product (\ref{innerprod-h}) can be naturally extended to quasi-cyclic codes of length $2m$ and index 2: 
\[ \big\langle  \big(a(x),b(x)\big) , \big(a'(x),b'(x)\big) \big\rangle _h =  
\big\langle  a(x),a'(x) \big\rangle _h + \big\langle  b(x), b'(x) \big\rangle _h .   \]

\begin{lemma} 
\label{gg-h} 
 If $f(x)$ divides $x^m-1$, then 
$$\left( \frac{x^m-1}{f(x)} \right)^{\dagger} = \frac{x^m(1-x^m)}{f(x)^{\dagger}}.$$ 

If the polynomial $g_{11}(x) g_{22}(x)$ divides $(x^m-1)g_{12}(x)$, then 
$$\left( \frac{(x^m-1)g_{12}(x)}{g_{11}(x)g_{22}(x)} \right)^{\dagger} = 
\frac{x^m(1-x^m) g_{12}(x)^{\dagger}}{g_{11}(x)^{\dagger}g_{22}(x)^{\dagger}} .$$
\end{lemma}

\begin{proof} 
Similar to the proof of Lemma \ref{gg}. 
\end{proof} 

\begin{theorem}
\label{dual-h}
Let $C$ be a quasi-cyclic code  of length $2m$ and index $2$, generated by two elements 
$g_1=\big( g_{11}(x),g_{12}(x)\big)$ and $g_2=\big( 0,g_{22}(x)\big)$, satisfying Conditions $(\ast)$. 
Then its Hermitian dual code $C^{{\perp}_h}$ is generated by 
two elements $\Big(  \frac{x^m(x^m-1)}{g_{11}(x)^{\dagger}} ,0 \Big)$ and 
$\Big( \frac{x^m(x^m-1)g_{12}(x)^{\dagger}} {g_{11}(x)^{\dagger}  g_{22}(x)^{\dagger} } ,  
- \frac{x^m(x^m-1)}{g_{22}(x)^{\dagger} } \Big)$. 
\end{theorem}

Theorem \ref{dual-h} can be reformulated as 

\begin{corollary}
\label{dual-cor-h}
Let $C$ be a quasi-cyclic code  of length $2m$ and index $2$, generated by two elements 
$g_1=\big( g_{11}(x),g_{12}(x)\big)$ and $g_2=\big( 0,g_{22}(x)\big)$, satisfying Conditions $(\ast)$. 
Then its Hermitian dual code $C^{{\perp}_h}$ is generated by 
two elements $\Big(  \frac{x^m-1}{\overline{g_{11}}(x)^{*}} ,0 \Big)$ and 
$\Big( \frac{x^{m+\deg g_{11} - \deg g_{12}}(x^m-1) \overline{g_{12}}(x)^{*}} {\overline{g_{11}}(x)^{*}  \overline{g_{22}}(x)^{*}} ,   
- \frac{x^m-1}{\overline{g_{22}}(x)^{*} } \Big)$. 
\end{corollary}

The next two theorems give characterizations of Hermitian  self-orthogonal and dual-containing codes. 

\begin{theorem}
\label{self-orth-h} 
Let $C$ be a quasi-cyclic code  of length $2m$ and index $2$, generated by elements 
$g_1=\big( g_{11}(x),g_{12}(x)\big)$ and $g_2=\big( 0,g_{22}(x)\big)$, satisfying Conditions $(\ast)$. 
Then $C$ is self-orthogonal with respect to the Hermitian inner product if and only if the following conditions hold:

1) $g_{22}(x) g_{22}(x)^{\dagger} \equiv 0  \pmod{x^m-1}$;

2) $g_{12}(x)g_{22}(x)^{\dagger} \equiv 0  \pmod{x^m-1}$;

3) $g_{11}(x)g_{11}(x)^{\dagger}  + g_{12}(x)g_{12}(x)^{\dagger} \equiv 0  \pmod{x^m-1}$. 
\end{theorem}

\begin{theorem}
\label{dual-con-h} 
Let $C$ be a quasi-cyclic code  of length $2m$ and index $2$, generated by elements 
$g_1=\big( g_{11}(x),g_{12}(x)\big)$ and $g_2=\big( 0,g_{22}(x)\big)$, satisfying Conditions $(\ast)$. 
Then $C$ is Hermitian dual-containing if and only if the following conditions hold:

1) $g_{11}(x) g_{11}(x)^{\dagger} $ divides  $x^m(x^m-1)$;

2) $g_{11}(x) g_{11}(x)^{\dagger}  g_{22}(x)$ divides  $x^m(x^m-1)g_{12}(x)$ ;

3) $g_{11}(x) g_{11}(x)^{\dagger}  g_{22}(x) g_{22}(x)^{\dagger}$ divides  
$x^m(x^m-1)(g_{11}(x)g_{11}(x)^{\dagger}  + g_{12}(x)g_{12}(x)^{\dagger} ) $. 
\end{theorem}

In terms of reciprocal polynomials, the previous theorems can be reformulated as 

\begin{corollary}
\label{self-orth-cor-h} 
Let $C$ be a quasi-cyclic code  of length $2m$ and index $2$, generated by elements 
$g_1=\big( g_{11}(x),g_{12}(x)\big)$ and $g_2=\big( 0,g_{22}(x)\big)$, satisfying Conditions $(\ast)$. 
Then $C$ is self-orthogonal with respect to the Hermitian inner product if and only if the following conditions hold:

1) $g_{22}(x) \overline{g_{22}}(x)^* \equiv 0  \pmod{x^m-1}$;

2) $g_{12}(x) \overline{g_{22}}(x)^* \equiv 0  \pmod{x^m-1}$;

3) $x^{\deg g_{12}} g_{11}(x) \overline{g_{11}}(x)^*  + x^{\deg g_{11}}  g_{12}(x) \overline{g_{12}}(x)^* \equiv 0  \pmod{x^m-1}$.
\end{corollary}

\begin{corollary}
\label{dual-con-cor-h} 
Let $C$ be a quasi-cyclic code  of length $2m$ and index $2$, generated by elements 
$g_1=\big( g_{11}(x),g_{12}(x)\big)$ and $g_2=\big( 0,g_{22}(x)\big)$, satisfying Conditions $(\ast)$. 
Then $C$ is Hermitian dual-containing if and only if the following conditions hold:

1) $g_{11}(x) \overline{g_{11}}(x)^*$ divides  $(x^m-1)$;

2) $g_{11}(x) \overline{g_{11}}(x)^*  g_{22}(x)$ divides  $(x^m-1)g_{12}(x)$;

3) $g_{11}(x) \overline{g_{11}}(x)^*  g_{22}(x) \overline{g_{22}}(x)^*$ divides  \\
\mbox{}  \hspace{0.8cm} $(x^m-1)(x^{\deg g_{22}} g_{11}(x) \overline{g_{11}}(x)^*  + 
 x^{\deg g_{11} + \deg g_{22} - \deg g_{12}}  g_{12}(x) \overline{g_{12}}(x)^* ) $. 
\end{corollary}

\section{Construction of quantum codes}
\label{quantum}

In this section we construct examples of quantum codes using 2-generator QC codes. 

We define $H^{\otimes n}$  to be the $n$-fold tensor product of the Hilbert space $H$ of dimension $q$  over the field of complex numbers $\mathbb{C}$.  Then $H^{\otimes n}$ is a Hilbert space of  dimension $q^n$. A quantum code of length $n$ and  dimension $k$ over $\mathbb F_q$ is defined to be a Hilbert subspace of  $H^{\otimes n}$ having  dimension $q^k$.  A quantum code with length $n$, dimension $k$ and minimum distance $d$ over $\mathbb F_q$ is denoted by $[[n, k, d]]_q$.


\begin{theorem}(\cite{Ashikhmin,Ketkar})   
\label{Q-s}
Let $C$ be a $[2m, k]$    linear code over $\mathbb F_q$, such that $ C\subseteq C^{{\perp}_s}$. 
If $k<m$, then there exists a quantum error-correcting code $Q$ with parameters $[[m, m-k, d_s]]_q$, 
where $d_s=\min\{w_s(c) \mid c \in C^{{\perp}_s} \setminus C \}$.
If $k=m$, then there exists a quantum error-correcting code $Q$ with parameters $[[m, 0, d_s]]_q$, 
where $d_s=\min\{w_s(c) \mid c \in C^{{\perp}_s} =C\}$.
\end{theorem}

We present examples of symplectic self-orthogonal codes and use them to construct  quantum codes.  
We used Magma \cite{Bosma} for our computations.

\begin{example} 
Let $m=45$, $q=2$, $x^m-1 = p_1 p_2 p_3 p_4 p_5 p_6 p_7 p_8$, where 
$p_1= x + 1$, $p_2= x^2 + x + 1$ $p_3= x^4 + x + 1$,  $p_4= x^4 + x^3 + 1$, $p_5= x^4 + x^3 + x^2 + x + 1$,  
$p_6= x^6 + x^3 + 1$,  $p_7= x^{12} + x^3 + 1$, and $p_8= x^{12} + x^9 + 1$.   

1) Let $g_{11}(x)= p_1 p_2 p_3$,  $g_{22}(x)= p_3 p_4 p_5 p_6 p_7 p_8$,   
$h(x)= x^{25} + x^{23} + x^{21} + x^{20} + x^{18} + x^{17} + x^{16} + x^{15} + x^{13} + x^{12} + x^6 + x^5 + x^3 + x^2$,  
and $g_{12}(x)= p_3 h(x)$. 
Then $C$ is a symplectic self-orthogonal code of dimension 41, $d_s(C)=13$, $d_s(C^{\perp_s})=11$, 
so $C$ determines a   $[[45,4,11]]$-qubit stabilizer code 
(which is a code with best known parameters 
among the  $[[45,4]]$ codes   \cite{Grassl}). 
The code $C^{\perp_s}$ is generated by elements $(g'_{11}(x), g'_{12}(x))$ and $(0, g'_{22}(x))$, where 
$g'_{11}(x)= (x^m-1)/g_{22}(x)^*= p_1p_2$, 
$g'_{22}(x)= (x^m-1)/g_{11}(x)^*= p_3 p_5 p_6 p_7 p_8$, and 
$g'_{12}(x) = x^{23}g_{12}(x)^*/\gcd(g_{11}(x)^*, g_{22}(x)^*) =x^{23}h(x)^*$.


2) Let   $g_{11}(x)= p_1 p_2 p_3$,  $g_{22}(x)= p_2 p_3 p_4 p_5 p_6 p_7 p_8$, 
$h(x)= x^{24} + x^{22} + x^{19} + x^{18} + x^{17} + x^{14} + x^{13} + x^{12} + x^4 + x^3 + x^2$,  and 
$g_{12}(x)= p_2 p_3 h(x)$. 
Then $C$ is a symplectic self-orthogonal code of dimension 39, $d_s(C)=14$, $d_s(C^{\perp_s})=10$, 
so $C$ determines $[[45,6,10]]$-qubit stabilizer code (which is a code with best known parameters 
among the $[[45,6]]$ codes \cite{Grassl}). 
\end{example} 

\begin{example} 
Let $m=18$, $q=2$,  $p_1(x) = x + 1$,  $p_2(x) = x^2 + x + 1$, $p_3(x) = x^6 + x^3 + 1$. 
Then $x^m-1 = (p_1(x) p_2(x) p_3(x))^2$. 
Let $g_{11}(x)= p_1(x) ^2 p_2(x) $,  $g_{22}(x)= p_1(x)  p_2(x) ^2 p_3(x) ^2$, 
$h(x)= x^9 + x^4 + x^3$,  and 
$g_{12}(x)= p_1(x)  p_2(x)  h(x)$. 
Then $C$ is a symplectic self-orthogonal code of dimension 12, $d_s(C)=8$, $d_s(C^{\perp_s})=5$, 
so $C$ determines an $[[18,3,5]]$-qubit stabilizer code (which is a code with best known parameters 
among the $[[18,3]]$ codes \cite{Grassl}). Note that in this case $\gcd(q,m) \ne 1$. 
\end{example}

\begin{example} 
The self-dual code $C$ from Example \ref{self-dual-s}  determines a $[[23,0,8]]$-qubit stabilizer code (which is a code with best known parameters among the $[[23,0]]$ codes \cite{Grassl}). 
Similarly, one can get  $[[24,0,8]]$, $[[25,0,8]]$, $[[26,0,8]]$-qubit stabilizer codes.  
\end{example}

\section{Conclusion}
\label{conclusion} 

 In this paper we have given a classification  of quasi-cyclic codes of index 2 over finite fields. 
 We presented a lower bound for their minimum distance. 
 A quasi-cyclic code of index 2  is generated  by at most two elements. 
We gave a necessary and sufficient condition for a quasi-cyclic code to be generated by one element. 
In addition, we gave a necessary and sufficient condition for the dual of a quasi-cyclic code to be generated by one element. 
We described duals of quasi-cyclic codes with respect to the Euclidean, symplectic and Hermitian inner products, with a uniform approach.  
Moreover, we described self-orthogonal and dual-containing codes with respect to these inner products. 
In our study we do not have restrictions on the characteristic of the ground field, except for the case of one-generator 
 codes. In particular, we do not use the decomposition of codes into a sum of component subcodes. We study codes using only standard generators. Finally, we used the codes obtained to  construct quantum error-correcting codes.

\bigskip 

{\bf Acknowledgments}	

\medskip
The first author is supported by UAEU (grant G00004233). 
The second author is supported by the Science Committee of the Ministry of Science 
and Higher Education of the Republic of Kazakhstan (Grant No. AP14869221). 
The third author is supported by Nanyang Technological University Research Grant No. 04INS000047C230GRT01. 
	

\end{document}